\AtBeginDocument{%
  \mathchardef\mathcomma\mathcode`\,
  \mathcode`\,="8000 
}
{\catcode`,=\active
  \gdef,{\mathcomma\discretionary{}{}{}}
}
\documentclass[runningheads]{llncs}
\usepackage{amssymb}
\usepackage{graphicx}
\usepackage{lineno}

\usepackage{soul, color}
\usepackage{amsmath}
\usepackage{tikz}
\usetikzlibrary[topaths, arrows, positioning]
\usepackage{cite}
\usepackage[ruled, vlined]{algorithm2e}
\usepackage{listings}
\begin{document}
\title{Preventive Model-based Verification and Repairing for SDN Requests
}
%
%
\author%
{%
Igor Burdonov\inst{1} \and 
Alexandre Kossachev\inst{1} \and 
Nina Yevtushenko \inst{1,2} \and 
Jorge L\'opez\inst{3,4} \and 
Natalia Kushik\inst{3} \and
Djamal Zeghlache \inst{3}
}
\authorrunning{I. Burdonov et al.}
%
\institute
{
Ivannikov Institute for System Programming of the Russian Academy of Sciences, Moscow, Russia \and
National Research University Higher School of Economics, Moscow, Russia \and
SAMOVAR, CNRS, T\'{e}l\'{e}com SudParis, Institut Polytechnique de Paris, \'{E}vry, France \and
Airbus Defense and Space, 1 Boulevard Jean Moulin, \'Elancourt, France
\email{\{igor,kos,evtushenko\}@ispras.ru, jorge.lopez-c@airbus.com, \{natalia.kushik,djamal.zeghlache\}@telecom-sudparis.eu}\\
}
\maketitle              
\begin{abstract}
Software Defined Networking~(SDN) is a novel network management technology, which currently attracts a lot of attention due to the provided capabilities. Recently, different works have been devoted to testing / verifying the (correct) configurations of SDN \emph{data planes}. In general, SDN forwarding devices (e.g., switches) route (steer) traffic according to the configured \emph{flow rules}; the latter identifies the set of virtual paths implemented in the data plane. In this paper, we propose a novel preventive approach for verifying that no misconfigurations (e.g., infinite loops), can occur given the requested set of paths. We discuss why such verification is essential, namely, how, when synthesizing a set of data paths, other \emph{not requested and undesired} data paths (including loops) may be unintentionally configured. Furthermore, we show that for some cases the requested set of paths cannot be implemented without adding such undesired behavior, i.e., only a superset of the requested set can be implemented. Correspondingly, we present a verification technique for detecting such issues of potential misconfigurations and estimate the complexity of the proposed method; its polynomial complexity highlights the applicability of the obtained results. Finally, we propose a technique for debugging and repairing a set of paths in such a way that the \emph{corrected} set does not induce undesired paths into the data plane, if the latter is possible.

\end{abstract}
\section{Introduction}\label{sec:intro}
Traditional networks have currently evolved. One of the technologies that contributes to this evolution is the Software Defined Networking (SDN) paradigm, that allows implementing various \emph{data paths} utilizing the common resources and control principles. When using SDN technology the network entities are managed through the controller that works independently of the network equipment and is `responsible' for pushing the necessary rules to the forwarding devices (e.g., switches) \cite{sdnchallenges}. As a result, SDN provides agile controllability and observability by separating the control and data planes. 

To guarantee the requested network is configured correctly, SDN components and compositions need to be thoroughly tested and verified. However, even if the rules are pushed to each switch as requested by the controller, additional verification of the data plane still needs to be performed. For example, one needs to verify i) the absence of loops and packet loss, as well as ii) the security and access control issues. The works on such data plane verification have been presented before (see Section~\ref{sec:related}), moreover, we note that these challenges have been largely investigated in the past decade. Nevertheless, existing approaches often rely on a current network configuration, i.e., the rules have been already pushed to the switches while in this paper, we claim that an efficient verification can be performed before. In particular, we propose to analyze the paths to be configured as it is highly probable that the loops and/or access control issues are not induced through the actual path implementation but rather arise from the conflicting user requests. 

More precisely, in this paper, we propose a novel preventive model-based approach for verifying certain network properties. Indeed, given the set of paths to be implemented on the data plane for connecting appropriate hosts, if this set is not consistent or can lead to potential loops then its implementation should be avoided. Let $P$ be a set of paths which should be implemented on the data plane for packets of a given traffic type. The set $P$ should be `inspected' before its actual implementation, first to assure that all the paths of the set $P$ are edge simple (proves the correctness of the path definition) and second whether it is possible to precisely implement the set $P$ on the data plane or there will be additional (unintended) paths implemented? In the latter case, it can happen that there are implemented paths which are not edge simple and thus, a loop for packets of a given traffic type can occur. In this paper, we answer the above question by establishing the corresponding necessary and sufficient conditions. In fact, we show that given a traffic type which is defined by the packet headers (packets with the same traffic type follow the same data paths) and a set of (requested) paths $P$, the implementation of $P$ can induce new paths appearing on the data plane, and moreover, if all the paths of $P$ are edge simple (no loops should occur) it does not guarantee the absence of potential cycles on the data plane. Indeed, the criterion for the absence of those relies on the property of the set $P$ to be arc closed (see Section~\ref{sec:impelementing}). Such criterion as well as the preventive verification method on its basis, form the main contributions of the paper. Note that, our preliminary experimental results with rather small topology built over the Onos controller and Open vSwitches confirm the necessity of such preventive verification; otherwise, the packets generated at a certain host can go into infinite loop and can simply flood the network. Together with the data path verification approach we also discuss the possibility of an automatic debugging and repairing of a set $P$ that did not pass the verification. The latter contribution of the paper is a technique for the modification of the set of paths $P$ in such a way, that the resulting set of paths becomes arc closed (and thus safe to implement). For both, verification and debugging / repairing approaches their related complexity is discussed.  

The structure of the paper is as follows. Section~\ref{sec:related} briefly summarizes the related work in the area of SDN data plane verification w.r.t. various network properties. Section~\ref{sec:prelim} presents the necessary background. Section~\ref{sec:impelementing} discusses the possibility of inducing undesired paths on the data plane that can cause, for example, infinite cycles. Correspondingly, the proposed preventive verification approach for the set of paths $P$, together with the criterion for the absence of undesired links and related complexity analysis is presented in Section~\ref{sec:checking}. Automatic debugging and repairing of the set of paths for which the verification failed, is proposed in Section~\ref{sec:debugging}. Section~\ref{sec:conc} concludes the paper.

\section{Related work}\label{sec:related}
A number of (recent) works have been devoted to verification and testing of an SDN data plane and related data paths configured on the data plane. Note that these works can be intuitively split into several groups. The first group focuses on the application of formal verification and model checking approaches to data plane verification or forwarding devices in isolation; in this case, \emph{classical} networks (not necessarily SDN) with related access control, security and other network properties are considered. Approaches of the second group tend to focus on \emph{active} testing of a data plane via corresponding traffic generation and monitoring of the forwarding behavior of switches of interest. There have been also a number of attempts of the application of model based testing techniques to various SDN components and in particular, to the data plane.

As techniques of the first group generally employ formal verification and model checking strategies, they mostly differ in the underlying formalism utilized for describing the specified behavior and related properties. For that matter, there have been considered Boolean functions and their satisfiability \cite{debugging}, symbolic model checking / execution and SMT solving \cite{canini2012nice, dobrescu2013toward} as well as algebra of sets \cite{verifheadersets}. Several properties of the data plane can be checked in this case, such as for example, reachability issues, absence of loops, etc. When verifying the behavior of forwarding devices in isolation, symbolic execution has been also employed. In fact, the problem can be solved via corresponding static analysis when the network device is implemented in the programming language (for example, P4) \cite{p4}. 

Approaches of the second group have been largely investigated, for example in \cite{autotestpacketgen,buzz,testdrivensdn,JJD19}. In automatic traffic generation, the packets / flows to be sent through the switches are generated at hosts in an active mode such that specific network failures can be captured when monitoring the data plane.

Existing model based testing techniques either consider a given SDN component, such as for example an SDN enabled switch \cite{yao2014formal, ICTSS18} or an SDN framework as a whole can be tested \cite{enase18, EWDTS18} and in this case, an appropriate fault model can be used / proposed.

Note that the authors are not aware of the (preventive) verification approaches applied to SDN when the specification is given as a set of paths to be implemented. Such verification should be performed beforehand, i.e., before the rules are pushed to the switches and at the same time, further network updates should be also verified not to bring undesired paths. On the other hand, we are not aware of any works devoted to data paths repairing in the context of SDN, and in this paper, we address the aforementioned challenges.

\section{Preliminaries}\label{sec:prelim}
Software Defined Networking~(SDN) is a networking paradigm that consists in separating the control and data plane layers \cite{sdn}. With a centralized SDN controller, SDN applications can automatically re-configure the SDN data plane. SDN-enabled forwarding devices (the components of the data plane) steer (route / forward) the incoming network packets based on so-called flow rules installed by the SDN applications (through the controller). A flow rule consists of three main (functional) parts: a packet matching part, an action part and a location / priority part. The matching part describes the values which a received network packet should have for a given rule to be applied. The action part states the required operations to perform to the matched network packets, while the location / priority part controls the hierarchy of the rules using tables and priorities. Finally, it is important to note that there exists a special output port for a flow rule, the controller port; when a packet is sent to the controller, the controller queries the SDN applications to decide the actions to perform to the packet; as a result, the controller may install new flow rules, drop or forward the packet to a specific port. In this paper, we focus on the resulting data paths (produced by the rules installed at the forwarding devices); more precisely, we focus on the analysis of such data paths and the potentially unintended additional data paths resulting from a configuration. To better outline the working principles of SDN rules, consider the following rules installed at a given switch:

\begin{center}
\begin{tabular}{c|c|c|c|c}
     \textbf{ID} & \textbf{Priority}  & \textbf{TCP DST PORT} & \textbf{DST IP} & \textbf{Action}\\ \hline
     1 & 5000 &     & 10.0.1.22 & OUT(2)    \\\hline
     2 & 5001 & 22  &           & OUT(3)    \\\hline
     3 & 6000 &     & 10.0.1.23 & CTRLLR    \\
\end{tabular}
\end{center}

To simplify our explanation, and without loss of generality we consider that the rules are installed in the first table of the SDN-enabled switch (table 0). TCP DST PORT is the TCP destination port and DST IP is the destination IP (for further information on basic networking concepts the reader can give a look at \cite{tcpip}). A network packet with the destination IP address 10.0.1.22 and destination TCP port 22 will be forwarded to the output 2 (due to the higher priority of rule 1). Likewise, a network packet with destination IP address 10.0.1.21 and destination TCP port 22 will be forwarded to port 3 (the highest priority rule matching the network packet). Finally, if a network packet going to the destination IP address 10.0.1.23 (and the destination TCP port not equal to 22) arrives, the switch sends this packet to the controller, asking for the action to take with the packet, the controller may reply with a new rule, drop the packet or forward it to a set of ports.

In this paper, the SDN \emph{resource} topology (data plane) or \emph{resource network connectivity topology (RNCT)} is represented as an undirected graph $G=(V,E)$ where $E\subseteq \{\{a,b\}| a\in V\;\&\;b\in V\}$ without multiple edges and loops. The set $V$ of nodes represents network devices such as \emph{hosts} and \emph{switches}; the set $H$ is the set of all hosts while $S$ is the set of all switches, $V=H\cup S$, $H\cap S=\emptyset$. Edges of the graph (the set $E$) represent connections (links) between two nodes in $G$ and each link can transmit packets in both directions. Correspondingly, given an edge between nodes $a,b\in E$, we write $(a, b)$ if a packet is transmitted from $a$ to $b$ and $(b, a)$ when it is transmitted from $b$ to $a$. We reasonably assume that each host is connected exactly with one switch, i.e., $\forall h\in H (\mathbf{deg}(h)=1) \;\&\; \exists s\in S  ((h,s)\in E)$ where $\mathbf{deg}(x)$ is the degree of the node $x$. Without loss of generality we also assume that $G$ is connected; otherwise, each (connected) component can be treated as a separate network. 

In the SDN architecture, the instructions for the data plane for packets’ forwarding are provided by SDN applications through an SDN-controller. These instructions (flow rules) produce so-called data paths, sets of paths which should carry on corresponding packets, i.e., those paths can have appropriate parameters according to which the packets are then forwarded; in other words, each packet belongs to an appropriate \emph{traffic type}. When a forwarding rule is installed on an SDN-enabled switch, a data link from and to other node (-s) adjacent to the switch is created, i.e., a packet accepted from adjacent nodes (hosts or switches) is forwarded to a (corresponding) set of ports that are connected to appropriate ports of other nodes.


A host can generate packets that are forwarded to a single switch connected with this host. A switch can only forward packets; moreover, in this paper, we assume that a switch does not modify the packet header, i.e., the packet’s traffic type and payload are not changed through the network. A switch can forward a packet to several ports, and the set of ports depends on the traffic type as well as on the input port from which it arrives. Every node $a$ of the graph $G$ (a host or a switch) has a set of \emph{ports} which can be input as well as output and each such port corresponds to some edge at the node $a$ and vice versa, each edge at the node $a$ is associated with a corresponding port. Thus, there is one-to-one correspondence between edges at the node $a$ and the set of its ports. Since $G$ has no multiple edges nor node (self) loops there is one-to-one correspondence between the set of ports of $a$ and the set of neighbor nodes of $a$. Therefore, without loss of generality, we can use a neighbor node instead of the port number. 

A path $\pi$ is a sequence of neighboring nodes of $G$,  i.e., a path is a sequence\footnote{As usual, we use `$\cdot$' for denoting the sequence concatenation.} of nodes such that there is an edge between neighboring sequence nodes. A path $\pi=x_1 \cdot \ldots \cdot x_n$ starts at the node $x_1$, is finished at the node $x_n$, has length $n-1$, and passes via an arc $(x_i, x_{i+1})$ for $i \in \{1, \ldots, n-1\}$. The path is \emph{edge simple} if it passes via each arc at most one time: $(x_i, x_{i+1}) = (x_j, x_{j+1}) \implies i = j$. The path is \emph{node simple} if all its nodes are pairwise different, i.e., $x_i = x_j \implies i = j$. A path is \emph{complete} if its head and tail nodes are hosts and there are no hosts as intermediate nodes.

An SDN application configures sets of paths (through the controller) which should transport corresponding packets, i.e., those paths can have appropriate parameters (which define their traffic type) according to which the packets are then forwarded \cite{ofspec}. The flow rules of a switch can be written as a mapping of input ports into subsets of output ports. If the subset of output ports is empty then the switch will `drop' a packet that arrived at a corresponding input port. 

In this paper, we assume that an SDN application configures the switch tables in such a way that each rule determines the set of output ports depending on the traffic type and an input port. As $G$ has no multiple edges it implies that a rule determines the set of neighboring nodes where a packet has to be forwarded. We also assume that all the switches have in their tables only the information sent by the controller, i.e., no default rules or external interfaces are considered. For the sake of simplicity and in fact, without loss of generality for our purpose, we assume that all the rules have the same priority. For packets belonging to the same traffic type, we can consider every rule as a triple $(a, s, b) \in V \times S \times V$ where $a$ and $b$ are neighbors of $s$. This rule says that getting a packet with the corresponding traffic type from neighbor $a$, switch $s$ should send it to the neighbor $b$. If there are several rules which differ only in the neighbor $b$, then switch $s$ performs \emph{cloning}, i.e., the incoming packet is transmitted to several neighbors. The set of rules of all switches is called \emph{configuration} (for the given traffic type). 

\section{Implementing the given set of complete paths}\label{sec:impelementing}
\subsection{Analysis of paths that can be implemented on the data plane}
The set of complete paths that should be implemented on the data plane is based on a user request or predefined configuration (by a given application). Correspondingly, before setting a switch configuration according to a set of paths, it would be useful to verify whether a given set of paths can be eventually implemented. Note that hereafter we assume that the requested set of paths $P$ does not contradict the RNCT $G$. A trivial check that $P$ forms a sub-graph of $G$ can be performed beforehand, if necessary.

When implementing a set of paths $P$, three options are possible. 1) $P$ can be implemented as it is and in this case, the edge simplicity should be verified for the set $P$. 2) $P$ cannot be implemented without implementing unintended paths, i.e., a superset of $P$ is implemented. In this case, the condition of the edge simplicity should be checked for this superset. If the minimal superset of $P$ that can exist on the data plane has cycling paths, then the set $P$ cannot be implemented (packet loops may flood the network) in the given data plane. 3) $P$ cannot be implemented but the minimal superset of $P$ that can be implemented satisfies the edge simplicity property.
We further discuss how given a set $P$ of paths, a corresponding switch configuration is specified and given a switch configuration, which paths are induced by this configuration. 

\paragraph{Complete paths induce switch rules} When implementing rules for a complete path (for the given traffic type) $\alpha \cdot a \cdot b \cdot c \cdot \beta$ where $a,b,c\in V, \alpha,\beta \in V^*$, we need a rule $(a, b, c)$, i.e., a switch $b$ once getting a packet belonging to this traffic type from the neighbor $a$ has to send it to the neighbor $c$. Formally, the set $P$ of paths induces the set $P\!\!\downarrow$ of rules: 

$\forall a\in V, b \in S, c\in V, \alpha \in V^*, \beta \in V^*$

$\alpha \cdot a \cdot b \cdot c \cdot \beta \in P$ implies that there is a rule	$(a, b, c) \in P\!\!\downarrow$.

\paragraph{Switch rules induce paths} The rule $(a, b, c)$ induces a path $a \cdot b \cdot c$ of length $2$. If there is a path $\alpha \cdot x \cdot y$ and there is a rule $(x, y, z)$ then there is a path $\alpha \cdot x \cdot y \cdot z$. Formally, a switch configuration $P\!\!\downarrow$ induces the set of complete paths, written $P\!\!\downarrow\uparrow$: 

$\forall b_j \in V$

$(a_1, b_1, b_2), (b_1, b_2, b_3), \ldots, (b_{n-1}, b_n, a_2) \in P\!\!\downarrow$ where $a_1$ and $a_2$ are the only hosts, 
there is a path $a_1 \cdot b_1 \cdot b_2 \ldots \cdot b_{n-1} \cdot b_n \cdot a_2$ in $P\!\!\downarrow\uparrow$.

By definition, the set $P\!\!\downarrow\uparrow$ has only complete paths. By the definition of $P\!\!\downarrow$ and $P\!\!\downarrow\uparrow$, the following statement holds.

\begin{proposition}\label{stmt:rulesinducepaths}
Given a switch $b$, for each rule $(a, b, c) \in P\!\!\downarrow$ of this switch, there is a path $\alpha \cdot a \cdot b \cdot c \cdot \beta \in P\!\!\downarrow\uparrow$ for some $\alpha$ and $\beta$. 
\end{proposition}

We now discuss the features of the set $P\!\!\downarrow\uparrow$. If there are two paths $\alpha \cdot x \cdot y \cdot \beta$ and $\alpha' \cdot x \cdot y \cdot \beta'$ in the set $P\!\!\downarrow\uparrow$ of complete paths, then according to the above rules, there are paths $\alpha \cdot x \cdot y \cdot \beta'$ and $\alpha' \cdot x \cdot y \cdot \beta$. Consider the case when $\alpha$ and $\alpha'$ are not empty, i.e., $x$ is a switch. If $\beta$ and $\beta'$ are not empty then according to the prefix of the path, switch $x$, once getting a packet passed the path $\alpha$ or the path $\alpha'$, sends the packet to switch $y$. According to the postfix, switch $y$, once getting a packet from switch $x$, sends it to the starting point of the paths $\beta$ and $\beta'$, and the packet passes the paths $\beta$ and $\beta'$. If $\beta$ and $\beta'$ are empty, then $y$ is a host and the packet passes both paths $\alpha \cdot x \cdot y$ and $\alpha' \cdot x \cdot y$ Therefore, the following statement holds.

\begin{proposition}\label{stmt:pathsinducepaths}
Given a switch configuration $P\!\!\downarrow$, $P\!\!\downarrow$ induces the set of complete paths $P\!\!\downarrow\uparrow$ with the following features:

$\forall \alpha, \alpha', \beta, \beta' \in V^*$

$\alpha \cdot x \cdot y \cdot \beta \in P\!\!\downarrow\uparrow \;\&\; \alpha' \cdot x \cdot y \cdot \beta' \in P\!\!\downarrow\uparrow$ $\implies$ $\alpha \cdot x \cdot y \cdot \beta' \in P\!\!\downarrow\uparrow$.
\end{proposition}

According to Proposition~\ref{stmt:pathsinducepaths}, the set of data paths on the data plane induced by the given set $P$ is exactly $P\!\!\downarrow\uparrow$, and in fact, it is the actual set of paths that gets implemented when requesting to implement the set $P$. 

The set $P$ of complete paths is \emph{closed with respect to a given arc} $(x, y)$ if for each two paths  $\alpha \cdot x \cdot y \cdot \beta$ and $\alpha' \cdot x \cdot y \cdot \beta'$ of the set $P$ which have a common arc $(x, y)$, paths $\alpha \cdot x \cdot y \cdot \beta'$ and  $\alpha' \cdot x \cdot y \cdot \beta$ are also in $P$. The set $P$ of paths is \emph{arc closed} if $P$ is closed w.r.t. each arc over the set $E$. Given a set $P$ of complete paths, the arc closure of $P$ is the smallest arc closed set of complete paths that contains $P$. 

According to the definition of an arc closed set and Proposition~\ref{stmt:pathsinducepaths}, the following statement can be established.

\begin{proposition}\label{stmt:arcclosure}
Given a set $P$ of complete paths, the set $P\!\!\downarrow\uparrow$ is the arc closure of $P$.
\end{proposition}

\begin{corollary}\label{stmt:inducedpathsconcide}
The set $P\!\!\downarrow\uparrow$ coincides with $P$ if and only if $P$ is arc closed. 
\end{corollary}

\begin{corollary}\label{stmt:simplepathsclosure}
If $P$ has only edge simple paths and is arc closed then $P\!\!\downarrow\uparrow$ has only edge simple paths.
\end{corollary}

According to Corollary~\ref{stmt:inducedpathsconcide}, the set $P$ can be implemented on the data plane (up to the equality relation) if and only if $P$ is arc closed, i.e., Corollary~\ref{stmt:inducedpathsconcide} establishes necessary and sufficient conditions for the precise implementation of set $P$ on the data plane (without additional `undesired' paths). 

If $P$ is not arc closed then $P$ cannot be implemented on the data plane (up to the equality relation). Moreover, sometimes $P$ cannot be implemented on the data plane at all as its arc closure has some cycling paths. Figure~\ref{fig:loops} shows an example when the set $P$ has two edge simple paths $\alpha$ and $\beta$ from initial host $h_0$ to the final host $h_1$ (left of the figure), the set of rules induced by this set is shown at the bottom and an induced path $\gamma$ of the set $P\!\!\downarrow\uparrow$ is illustrated at the right. The path is not edge simple, and this example illustrates that cycles can occur even when paths of the set $P$ are simple.

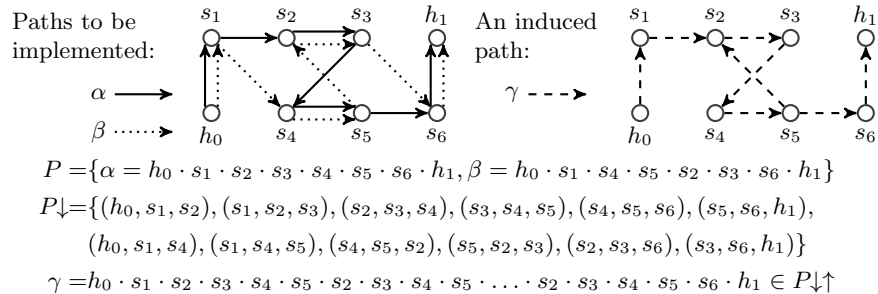
\begin{figure*}[!htb]
    \centering
    \begin{tikzpicture}[node distance=1cm,>=stealth',bend angle=45,auto]
            \tikzstyle{node}=[circle,thick,draw=black!75, inner sep=0.075cm]
            \tikzstyle{normalArrow}=[thick,->]
            \tikzstyle{dashedArrow}=[thick,->,dashed]
            \tikzstyle{dottedArrow}=[thick,->,dotted]

            \node[node]                 (s1a)       {};
            \node[]                     (s1atext)   at ([yshift=0.3cm]s1a.center)  {$s_1$};
            \node[node, right of=s1a]   (s2a)       {};
            \node[]                     (s2atext)   at ([yshift=0.3cm]s2a.center)  {$s_2$};
            \node[node, right of=s2a]   (s3a)       {};
            \node[]                     (s3atext)   at ([yshift=0.3cm]s3a.center)  {$s_3$};
            \node[node, right of=s3a]   (h1a)       {};
            \node[]                     (h1atext)   at ([yshift=0.3cm]h1a.center)  {$h_1$};
            \node[node, below of=s1a]   (h0a)       {};
            \node[]                     (h0atext)   at ([yshift=-0.3cm]h0a.center)  {$h_0$};
            \node[node, right of=h0a]   (s4a)       {};
            \node[]                     (s4atext)   at ([yshift=-0.3cm]s4a.center)  {$s_4$};
            \node[node, right of=s4a]   (s5a)       {};
            \node[]                     (s5atext)   at ([yshift=-0.3cm]s5a.center)  {$s_5$};
            \node[node, right of=s5a]   (s6a)       {};
            \node[]                     (s6atext)   at ([yshift=-0.3cm]s6a.center)  {$s_6$};
            
             \path  (h0a.north east)    edge[dottedArrow]   node[] {}   (s1a.south east)
                    (s1a)               edge[dottedArrow]   node[] {}   (s4a)
                    (s4a.south east)    edge[dottedArrow]   node[] {}   (s5a.south west)
                    (s5a)               edge[dottedArrow]   node[] {}   (s2a)
                    (s2a.south east)    edge[dottedArrow]   node[] {}   (s3a.south west)
                    (s3a)               edge[dottedArrow]   node[] {}   (s6a)
                    (s6a.north east)    edge[dottedArrow]   node[] {}   (h1a.south east)
             ;
             
             \path  (h0a.north west)    edge[normalArrow]    node[] {}   (s1a.south west)
                    (s1a)               edge[normalArrow]   node[] {}   (s2a)
                    (s2a.north east)    edge[normalArrow]   node[] {}   (s3a.north west)
                    (s3a)               edge[normalArrow]   node[] {}   (s4a)
                    (s4a.north east)    edge[normalArrow]   node[] {}   (s5a.north west)
                    (s5a)               edge[normalArrow]   node[] {}   (s6a)
                    (s6a.north west)    edge[normalArrow]   node[] {}   (h1a.south west)
             ;
             
             \node[text width=2.3cm]    (legend1)   at ([xshift=-1.5cm]s1a)         {Paths to be implemented:};
             \node[below of=legend1]    (alphatext) at ([yshift=0.25cm]legend1)     {$\alpha$};
             \node[]                    (alphas)    at ([xshift=0.1cm]alphatext)    {};
             \node[right of=alphas]     (alphae)                                    {};
             \node[below of=alphatext]  (betatext)  at ([yshift=0.5cm]alphatext)    {$\beta$};
             \node[]                    (betas)     at ([xshift=0.1cm]betatext)     {};
             \node[right of=betas]      (betae)                                     {};
             
             \path (alphas)             edge[normalArrow]   node[] {}   (alphae);
             
             \path (betas)             edge[dottedArrow]   node[] {}   (betae);
             
             \node[text width=2cm]      (legend2)   at ([xshift=1.5cm]h1a)         {An induced path:};
             \node[below of=legend2]    (gammatext) at ([yshift=0.25cm,xshift=-0.5cm]legend2)     {$\gamma$};
             \node[]                    (gammas)    at ([xshift=0.1cm]gammatext)    {};
             \node[right of=gammas]     (gammae)                                    {};
             
             \path (gammas)             edge[dashedArrow]   node[] {}   (gammae);
             
            \node[node]                 (s1b)       at ([xshift=1.2cm]legend2)       {};
            \node[]                     (s1btext)   at ([yshift=0.3cm]s1b.center)  {$s_1$};
            \node[node, right of=s1b]   (s2b)       {};
            \node[]                     (s2btext)   at ([yshift=0.3cm]s2b.center)  {$s_2$};
            \node[node, right of=s2b]   (s3b)       {};
            \node[]                     (s3btext)   at ([yshift=0.3cm]s3b.center)  {$s_3$};
            \node[node, right of=s3b]   (h1b)       {};
            \node[]                     (h1btext)   at ([yshift=0.3cm]h1b.center)  {$h_1$};
            \node[node, below of=s1b]   (h0b)       {};
            \node[]                     (h0btext)   at ([yshift=-0.3cm]h0b.center)  {$h_0$};
            \node[node, right of=h0b]   (s4b)       {};
            \node[]                     (s4btext)   at ([yshift=-0.3cm]s4b.center)  {$s_4$};
            \node[node, right of=s4b]   (s5b)       {};
            \node[]                     (s5btext)   at ([yshift=-0.3cm]s5b.center)  {$s_5$};
            \node[node, right of=s5b]   (s6b)       {};
            \node[]                     (s6btext)   at ([yshift=-0.3cm]s6b.center)  {$s_6$};
            
             \path  (h0b)   edge[dashedArrow]   node[] {}   (s1b)
                    (s1b)   edge[dashedArrow]   node[] {}   (s2b)
                    (s2b)   edge[dashedArrow]   node[] {}   (s3b)
                    (s3b)   edge[dashedArrow]   node[] {}   (s4b)
                    (s4b)   edge[dashedArrow]   node[] {}   (s5b)
                    (s5b)   edge[dashedArrow]   node[] {}   (s2b)
                    (s5b)   edge[dashedArrow]   node[] {}   (s6b)
                    (s6b)   edge[dashedArrow]   node[] {}   (h1b)
             ;
             
             \node[text width=12cm]    (legend3)  at ([yshift=-1cm]s6atext)
             {
                \begin{align*}
                P=&\{\alpha=h_0 \cdot s_1 \cdot s_2 \cdot s_3 \cdot s_4 \cdot s_5 \cdot s_6 \cdot h_1, \beta= h_0 \cdot s_1 \cdot s_4 \cdot s_5 \cdot s_2 \cdot s_3 \cdot s_6 \cdot h_1\}\\
                P\!\!\downarrow=&\{(h_0,s_1,s_2),(s_1,s_2,s_3),(s_2,s_3,s_4),(s_3,s_4,s_5),(s_4,s_5,s_6),(s_5,s_6,h_1),\\
                                &(h_0,s_1,s_4),(s_1, s_4, s_5),(s_4,s_5,s_2),(s_5,s_2,s_3),(s_2,s_3,s_6),(s_3,s_6,h_1)\}\\
                \gamma=&h_0 \cdot s_1 \cdot s_2 \cdot s_3 \cdot s_4 \cdot s_5 \cdot s_2 \cdot s_3 \cdot s_4 \cdot s_5 \cdot \ldots \cdot s_2 \cdot s_3 \cdot s_4 \cdot s_5 \cdot s_6 \cdot h_1 \in P\!\!\downarrow\uparrow
                \end{align*}
             };
    \end{tikzpicture}
    \caption{Induced (cyclic) paths’ occurrence}
    \label{fig:loops}
\end{figure*}

Similar to $P$, all the paths of the set $P\!\!\downarrow\uparrow$ are complete paths. However, if  $P\!\!\downarrow\uparrow$ is a proper superset of $P$ then we have to check whether all the paths of the set $P\!\!\downarrow\uparrow$ are edge simple. If it is the case then the set $P$ can be implemented on the data plane up to the set $P\!\!\downarrow\uparrow$ (i.e., with additional unspecified paths from $P\!\!\downarrow\uparrow \setminus P$). If it is not the case then the set $P$ should be modified and this issue is discussed in Section~\ref{sec:debugging}. 

From the practical point of view, perhaps the most interesting application is when some set $P\!\!\downarrow\uparrow$ of paths is already implemented on the data plane and a new request arrives; either a request $A$ to add new paths ($P\cup A$) or a request $R$  to remove paths ($P\setminus R$) to / from the original set. In this case, the same check should be performed on $((P\cup A)\setminus R)\!\!\downarrow\uparrow$ before implementing / removing paths, guaranteeing the implementability of the augmented set of paths. Algorithm~\ref{algo:checking}  summarizes the necessary verification steps (Section~\ref{sec:checking}) and returns the corresponding verdict about the implementability of a given set of paths.

\subsection{Practical / Experimental motivation}
It is worth noting that though the approach presented above is theoretical, the implications for real SDN frameworks are substantial. Indeed, if two loopless paths can induce (infinitely) more paths, the performance and security of such frameworks can be highly compromised. In order to verify if our (fundamental) findings can occur in real SDN framework implementations, an experimental evaluation was performed. 

Experiments were carried in a virtual machine running GNU/Linux CentOS 7.6 with 8 vCPUs and 16GB of RAM. The Onos \cite{berde2014onos} SDN controller (version 4.2.8) was installed via a Docker \cite{docker} container. To emulate the SDN data-plane, the Containernet \cite{containernet} was also installed through a Docker container.


The paths shown in Figure~\ref{fig:loops} were configured independently, successful communication from $h_0$ to $h_1$ was discovered using the data path discovery tool presented in \cite{JJD19} and the discovered paths are shown in Figure~\ref{fig:independent_paths}. As can be seen, there is no problem while configuring both paths independently. When both paths were configured simultaneously, the loop was effectively produced. A single packet sent from $h_0$ to $h_1$ produced infinitely many of them. In Figure~\ref{fig:pcap}, we show the packet dump (using the well-known utility \lstinline[language=bash]{tcpdump}) as seen by $h_1$. Note that, the packet sent is an ICMP echo request (using the \lstinline[language=bash]{ping} utility), and the sequence ID is always 1, as the single packet gets copied infinitely many times. When continuously sending the packets the network rapidly degraded until the whole infrastructure became unusable. 

\begin{figure}
    \centering
    \begin{tabular}{c}
        \includegraphics[angle=90, width=\textwidth]{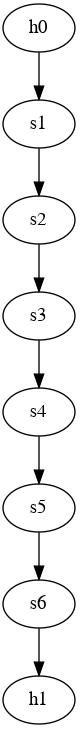} \\
        $\alpha$\\
        \includegraphics[angle=90, width=\textwidth]{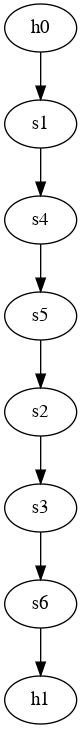} \\
        $\beta$\\
    \end{tabular}
    \caption{Discovered data-paths ($\alpha$ and $\beta$)}
    \label{fig:independent_paths}
\end{figure}

\begin{figure}
    \centering
    \includegraphics[width=\textwidth]{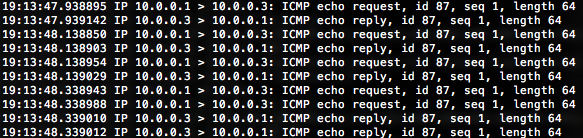}
    \caption{Packet capture showing an infinite loop in the experimental infrastructure}
    \label{fig:pcap}
\end{figure}

These experiments confirm the importance of our findings. Indeed, it is important to provide SDN frameworks with verification tools before rules are pushed to the switches. One of the procedures for such verification is given in Algorithm~\ref{algo:checking}. Note that,  Corollary~\ref{stmt:inducedpathsconcide} provides a criterion for effective verification of the set of paths $P$. However, for that matter the $P\!\!\downarrow\uparrow$ (the arc closure of $P$) needs to be derived as well, and this issue is discussed in the next section. 

\section{Checking the arc closure}\label{sec:checking}
In this section, we propose an algorithm for checking if a given set of paths $P$ induces unintended paths, i.e., a superset of $P$ is implemented (when $P$ is intended); likewise, we discuss how to detect potential cycles induced by the implementation of $P$.

Algorithm~\ref{algo:checking} shows the verification steps necessary to check the arc closure of a given set of paths. Given the set $P$ of complete paths in the graph $G$, we construct a directed graph $D(P)$. Vertices of $D(P)$ are arcs of paths from $P$ and there is an arc $((a, b), (b, c))$ in $D(P)$ if and only if $P$ has a path $\alpha \cdot a \cdot b \cdot c \cdot \beta$ where $\alpha$ and $\beta$ are not empty sequences. There are two special nodes in $D(P)$, the initial node $source$, and the final node $sink$. Since $P$ contains only complete paths, in the graph $D(P)$, there is an edge from the $source$ vertex to a head pair $(a, b)$ of each path where $a$ is a host, while there is an edge to the $sink$ node from the tail pair $(c, d)$ of each path, where $d$ is a host. The path $source \cdot (h_1, s_2) \cdot (s_2, s_3) \cdot \ldots \cdot (s_{m - 1}, h_m) \cdot sink$ in the graph $D(P)$ starting at the $source$ vertex and ending at the $sink$ vertex corresponds to the complete path $h_1 \cdot s_2 \cdot s_3 \cdot \ldots \cdot s_{m - 1} \cdot h_m$ in the graph $G$ where $h_1$ and $h_m$ are hosts. The set of such complete paths in the graph $G$, corresponding to the paths in the graph $D(P)$ from $source$ to $sink$, is precisely the closure of the set $P$. If the number of such paths in $D(P)$ is greater than the cardinality of the set $P$, this means that the closure expands the set $P$. The detailed verification procedure is shown in Algorithm~\ref{algo:checking} and Proposition~\ref{proposition_checking} (valid by construction) establishes the correctness of the algorithm. Note that the algorithm always terminates due to the finite calculations in nested loops, independently if $P\!\!\downarrow\uparrow$ contains a path with a loop or not.

\begin{algorithm}[!htb]
\SetKwInOut{Input}{Input}\SetKwInOut{Output}{Output}\SetKw{KwBy}{by}
    \Input{A set $P$ of edge simple complete paths}
    \Output{A verdict whether the set $P$ is arc closed}
    
    Derive a subset $Q = \{q_1, \dots, q_k\}$ of $P$ that contains all the paths of length greater than two; we denote as $k_j$ the length of a path $q_j$, $j \in \{1, \dots, k\}$\;
    
    Derive a graph $D(P) = <D, E>$ for the set $Q$ where the vertices of $D(P)$ are pairs of vertices of the paths in $Q$\;
    
    $D = \{source, sink\}$; $E = \emptyset$\;
    $j = 0$\;

    \While {$j < k$}
    {
        $j++$;
        $D = D \cup \{(q_j(1), q_j(2)), (q_j(k_j), q_j(k_j+1))\}$\; 
        $E = E \cup \{(source, (q_j(1), q_j(2)), (q_j(k_j), q_j(k_j+1)), sink)\}$\; 
        $m = 2$\;

        \While {$m < k_j + 1$}
        {
            $D = D \cup \{(q_j(m), q_j(m + 1))\}$\; 
            $E = E \cup \{((q_j(m-1), q_j(m)), (q_j(m), q_j(m + 1)))\}$\;
            $m++$\;
        }
    }
    \If{the number of paths in $D(P)$ from $source$ to $sink$ is greater than $k$} 
    {
        \Return{$False$}\;
    }
    
    \Return{$True$}\;
    
\caption{Verifying if the set of paths $P$ is arc closed} \label{algo:checking}
\end{algorithm}

\begin{proposition} \label{proposition_checking}
Algorithm~\ref{algo:checking} returns the verdict $True$ if and only if $P$ is arc closed.
\end{proposition}

Consider the example in Figure~\ref{fig:loops}, the graph $D(P)$ constructed by Algorithm~\ref{algo:checking} is the following. The set of vertices is $\{source, (h_0, s_1), (s_1, s_2), (s_1, s_4), (s_2, s_3), (s_5, s_2), (s_3, s_4), (s_3, s_6), (s_4, s_5), (s_5, s_6), (s_6, h_1), sink\}$ and the corresponding graph is shown in Figure~\ref{fig:D(P)}. By direct inspection one can assure that there is a cycle $(s_2, s_3), (s_3, s_4), (s_4, s_5), (s_5, s_2)$ in the graph and thus, the number of paths from the vertex $source$ to the vertex $sink$ is infinite, i.e., is bigger than the number two of paths in the set $P$, and therefore, the set $P$ is not arc closed as it is demonstrated in Figure~\ref{fig:loops}.

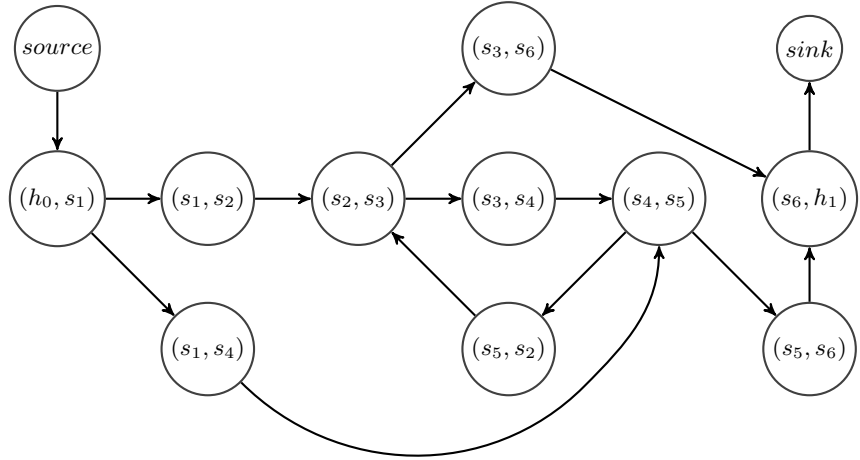
\begin{figure*}[!htb]
    \centering
    \begin{tikzpicture}[node distance=2cm,>=stealth',bend angle=45,auto]
            \tikzstyle{node}=[circle,thick,draw=black!75, inner sep=0.075cm]
            \tikzstyle{normalArrow}=[thick,->]
           
            \node[node]                 (source)    {$source$};
            \node[node, below of=source] (h0s1)     {$(h_0,s_1)$};
            \node[node, right of=h0s1]   (s1s2)       {$(s_1,s_2)$};
            \node[node, right of=s1s2]   (s2s3)       {$(s_2,s_3)$};
            \node[node, below of=s1s2]   (s1s4)       {$(s_1,s_4)$};
            \node[node, right of=s2s3]   (s3s4)       {$(s_3,s_4)$};
            \node[node, above of=s3s4]   (s3s6)       {$(s_3,s_6)$};
            \node[node, below of=s3s4]   (s5s2)       {$(s_5,s_2)$};
            \node[node, right of=s3s4]   (s4s5)       {$(s_4,s_5)$};
            \node[node, right of=s4s5]  (s6h1)      {$(s_6,h_1)$};
            \node[node, above of=s6h1]   (sink)       {$sink$};
            \node[node, below of=s6h1]   (s5s6)       {$(s_5,s_6)$};

             \path  (source) edge[normalArrow]  (h0s1) 
                    (h0s1) edge[normalArrow]  (s1s4)
                    (h0s1) edge[normalArrow]  (s1s2)
                    (s1s2) edge[normalArrow]  (s2s3)
                    (s2s3) edge[normalArrow]  (s3s6)
                    (s2s3) edge[normalArrow]  (s3s4)
                    (s3s4) edge[normalArrow]  (s4s5)
                    (s3s6) edge[normalArrow]  (s6h1)
                    (s5s6) edge[normalArrow]  (s6h1)
                    (s4s5) edge[normalArrow]  (s5s6)
                    (s4s5) edge[normalArrow]  (s5s2)
                    (s5s2) edge[normalArrow]  (s2s3)
                    (s6h1) edge[normalArrow]  (sink)
             ;
             \coordinate (out) at ([xshift=1cm,yshift=-0.5cm]s5s2);
             \draw [thick]  (s1s4) to [out=-45,in=225] (out) [right,thick,->,in=225,out=45] to [in=-90] (s4s5);
    \end{tikzpicture}
    \caption{Graph $D(P)$ for verifying the set of paths $P$}
    \label{fig:D(P)}
\end{figure*}

\begin{proposition} \label{proposition_checking_complexity}
The complexity of checking the absence of cycles for a given set of paths $P$ is $\mathcal{O}(L + |V|^3)$ where $|V|$ is the number of nodes in $G$ and $L$ is the sum of the lengths of the paths in $P$.
\end{proposition}
\begin{proof}
The complexity of constructing the graph $D(P)$ is $\mathcal{O}(L)$ where $L$ is the sum of the lengths of the paths from $P$.
In order to check for (infinite) loops, the absence of oriented cycles in the graph $D(P)$ needs to be checked, which is done through a topological sort (e.g., using depth first search (DFS) 
\cite{introductiontoalgorithms}). DFS-algorithm can also be used for computing the number of paths from the $source$ to the $sink$ node when there are no cycles. The running time of the depth first search algorithm on the graph $D(P)$ is evaluated as $\mathcal{O}(m)$, where $m$ is the number of arcs of the graph $D(P)$, $m \leq |V|^3$.
\end{proof}

\section{Debugging and Repairing a set of paths}\label{sec:debugging}
In this section, we discuss some possibilities of correcting / modifying the set of paths $P$ whenever this set is not arc closed. One first needs to identify the reason, i.e., a subset of paths that destroy the corresponding property, and the set of paths $P$ should be either augmented with new paths or on the contrary, certain paths should be deleted from the set $P$. In both ways, the resulting subset becomes arc closed and thus, can be implemented on the data plane without any additional links. We later on refer to this process as automatic $P$ \emph{debugging} and \emph{repairing}. We note, that such repairing process can have various objectives, such as for example: minimization of the number of paths to be excluded / included from / to $P$, maximization of a host to host connectivity in the resulting set of paths, minimization of the number of changes in the paths of the set, minimization of virtual links on the data plane, etc. We furthermore discuss some of the possibilities listed above and propose various debugging and repairing strategies.

\subsection{Minimizing the set of paths to be excluded / included from / to $P$}
Given a set $P$ of complete paths, let $P = \{p_1, \dots, p_k\}$, i.e., $k = |P|$, and $k_i = |p_i|-1$, i.e., $k_i$ is the length of $p_i$ for all $i \in \{1, \dots, k\}$. The problems we address in this subsection are the following: how to delete / add a minimal number of paths from / to the set $P$, such that the resulting subset / superset becomes arc closed. 

We say that two different paths $p_i$ and $p_j$ of $P$ are \emph{incompatible} if there exists a common arc, i.e., there exist $u \in \{1, \dots, k_i-1\}$ and $v \in \{1, \dots, k_j-1\}$ such that $p_i(u) = p_j(v)$ $\&$ $p_i(u + 1) = p_j(v + 1)$ while a path $p_i(1) \cdot \ldots \cdot p_i(u) \cdot p_j(v+1) \cdot \ldots \cdot p_j(k_j+1)$ or a path $p_j(1) \cdot \ldots \cdot p_j(v) \cdot p_i(u+1) \cdot \ldots \cdot p_i(k_i+1)$ is not in $P$. In this case, one can also say that $p_i$ and $p_j$ are incompatible w.r.t. the common arc $(a, b) = p_i(u), p_i(u+1)$. If $p_i$ and $p_j$ of $P$ are not incompatible, then they are \emph{compatible}.

The problem of deleting a minimal number of paths can be reduced to the well known maximum independent set problem. For that matter, we propose to derive an un-directed graph $G(P)$ in the following way: the nodes of the graph correspond to the paths of the set $P$. There is an arc between $p_i$ and $p_j$, $i \neq j$, in the graph $G(P)$ if the paths $p_i$ and $p_j$ are incompatible.

Given an un-directed graph $G(P)$, note that a subset of nodes which are not pairwise connected is an independent subset of nodes. Therefore, by construction, the following proposition holds.
\begin{proposition} \label{G(P)_proposition}
An independent subset of nodes of graph $G(P)$ is an arc closed set.
\end{proposition}

\begin{corollary}
A subset of $P$ is arc closed if and only if it is an independent subset of the graph $G(P)$.
\end{corollary}

Therefore, the problem of minimizing the set of paths to be excluded from $P$ is reduced to the derivation of a maximal independent subset of nodes in $G(P)$. Note that this problem is known to be NP-hard, and thus the repairing approach can be more complex than that one presented for the verification itself (Section~\ref{sec:checking}).

As an example, consider again the paths of the set $P$ in Figure~\ref{fig:loops}. Note that the paths from $P$ possess the necessary feature, i.e., they have a common arc $(s_2, s_3)$ with the above property and the set $P$ has no path $h_0 \cdot s_1 \cdot s_2 \cdot s_3 \cdot s_4 \cdot s_5 \cdot s_2 \cdot s_3 \cdot s_4 \cdot s_5 \cdot s_6 \cdot h_1$. Therefore, the corresponding vertices in $G(P)$ are connected, i.e., $P$ is not arc closed and only the singletons $\{\alpha\}$ or $\{\beta\}$ are arc closed.

For deriving a minimal superset of $P$ that is arc closed, the graph $D(P)$ derived in the previous subsection can be used. If the graph returned by Algorithm~\ref{algo:checking} has no cycles then the set of all paths from the $source$ node to the $sink$ node is the smallest superset of $P$ that is arc closed. Correspondingly, the following statement holds.

\begin{proposition} 
1. If all the paths from the $source$ node to the $sink$ node in $G(P)$ are edge simple then the set of all paths is the smallest superset of $P$ that is arc closed. 2. If there a path from the $source$ node to the $sink$ node in $G(P)$ that is not edge simple then there is no finite superset of $P$ that is arc closed. 
\end{proposition}

Note that in case 2, it is not possible to add paths to the given set $P$; the set $P$ can be only reduced as it is discussed at the beginning of the subsection. Indeed, it is exactly the case for the set $P$ in Figure~\ref{fig:loops}.

\subsection{Minimizing the number of arc changes in the set $P$}
Consider a set $P$ of edge simple complete paths that is not arc closed, the question arises: can the paths of the set be minimally corrected (w.r.t. the number of arcs) in order to get an arc closed set preserving the head and tail hosts of each path? In this section, we propose a simple way for modifying a single edge or a sub-path of a path using edges of the RNCT graph $G$ which were not utilized in the paths of $P$ (the set $N$ in Algorithm~\ref{algo:debugging}). 

\begin{algorithm}[!htb]
\SetKwInOut{Input}{Input}\SetKwInOut{Output}{Output}\SetKw{KwBy}{by}
    \Input{A set $P$ of edge-simple complete paths that is not arc closed, a non-empty set $N$ of edges between switches of the RNCT graph $G$ which are not used in the paths of the set $P$}
    \Output{A verdict $False$ if paths cannot be modified, or a modified arc closed set $P$ where the head and tail vertices of each modified path $p_j'$ coincide with those of the initial path $p_j$ of $P$}
    
    $Q = \{p_1\}$\;
    $j = 2$\; 

    \While {$j \leq |P|$}
    {
        $p_j' = p_j$\; 
        $l = 1$\; 

        \While {$l \leq |Q|$}
        {
            $p=q_l$\;
            \If {paths $p$ and $p_j'$ are incompatible w.r.t. $P$} 
            {   
                \If {$N = \emptyset$}
                {
                    \Return $False$\;
                }
                \Else 
                {
                    \While {the paths $p$ and $p_j'$ are incompatible w.r.t. the common arc $(s_1, s_2)$}
                    {
                        \uIf {the paths $p$ and $p_j'$ have a common sub-path $s_3 \cdot \alpha \cdot s_1 \cdot s_2 \cdot \beta \cdot s_4$ and $(s_3, s_4)$ is in $N$}
                        {
                            Derive $p_j'$ by replacing a sub-path $s_3 \cdot \alpha \cdot s_1 \cdot s_2 \cdot \beta \cdot s_4$ in $p_j'$ by a sub-path $s_3 \cdot s_4$\;
                        Delete $(s_3, s_4)$ from the set $N$\;
                        }
                        \uElseIf {There is a switch $s_3$ such that $(s_1, s_3), (s_3, s_2) \in N$ }
                        {
                            Derive $p_j'$ by replacing a sub-path $s_1 \cdot s_2$ in $p$ by a sub-path $s_1 \cdot s_3 \cdot s_2$\;
                            Delete $(s_1, s_3)$ and $(s_3, s_2)$ from the set $N$\; 
                        }
                        \uElse 
                        {
                            \Return $False$\;
                        }
                    }
                }
           }
           $l++$\;
        }
        Add $p_j'$ to the set $Q$\; 
        $j++$\;
    }

\Return an arc closed set $Q=\{p_1, p_2', \dots, p_k'\}$
\caption{Repairing via modifying an edge or a sub-path preserving the head and tail hosts of the path} \label{algo:debugging}
\end{algorithm}

By construction, the following statement holds.
\begin{proposition} \label{proposition_minimizing_arcs}
Given a set $P$ of edge-simple complete paths, if Algorithm~\ref{algo:debugging} returns a set $Q = \{p_1, p_2', \dots, p_k'\}$ then this set is arc closed and for each $j \in \{1, \dots, k\}$, the head and tail vertices of $p_j'$ coincide with those of $p_j$. 
\end{proposition}

Note that the set of repaired paths returned by Algorithm~\ref{algo:debugging} has only edge simple paths, since every time only unused links are utilized for the replacement. For the same reason, this set is arc closed. Moreover, we consider only simple heuristics for repairing a path; note as well that the result significantly depends on the order of the paths in $P$. More research is needed to propose more rigorous conditions for repairing a set of initial paths that is not arc closed. Those conditions can be related to certain properties as the link load distribution and thus, could re-direct some packets, for example, for traffic optimization.

As an example, consider again the paths in Figure~\ref{fig:loops}, assuming that each pair of switches is connected in the RNCT $G$. These paths have a common arc $(s_4, s_5)$ that can be replaced by a path $s_4 \cdot s_6 \cdot s_5$. After this modification the paths have a common arc $(s_2, s_3)$ that can be replaced by a path $s_2 \cdot s_4 \cdot s_3$. Thus, we obtain an arc closed set of paths $P' = \{h_0 \cdot s_1 \cdot s_2 \cdot s_3 \cdot s_4 \cdot s_5 \cdot s_6 \cdot h_1, h_0 \cdot s_1 \cdot s_4 \cdot s_6 \cdot s_5 \cdot s_2 \cdot s_4 \cdot s_3 \cdot s_6 \cdot h_1\}$.

\section{Conclusion}\label{sec:conc}
In this paper, we discussed some implementability issues for a given set of paths on an SDN data plane. We showed that for a fixed traffic type, whenever the requested set contains only edge simple paths, more (unintended) paths can still be implemented on the data plane, and some of those can create cycles, i.e., infinite packet loops. We therefore established the necessary and sufficient conditions for a set of requested paths to be implemented without any undesired connections and hence, potential loops. Our preventive verification approach is based on the analysis of the set of paths to be arc closed that in fact guarantees its `clean' (exact) implementability; this can be useful for guaranteeing that new (requested) and preexisting paths form valid configurations. The estimated (polynomial w.r.t. the total paths’ length) complexity of the proposed approach makes believing in its applicability for large scale virtual networks. At the same time, for a set of paths that cannot be implemented directly on the data plane, we proposed a debugging and repairing approaches for correcting the initial request, such that the resulting set becomes arc closed.

As future work, we plan to extend the proposed approaches abstracting from a given traffic type, i.e., considering sets of paths that share certain parameters of the packet header. Complexity issues in this case form maybe the main challenge, and thus we plan to study certain properties of various headers’ partitioning to check the implementability of a given set of paths. Moreover, it can be interesting to consider other kinds of specifications for user requests, such as for example, given pairs of hosts to be connected on the data plane, one needs to face the implementability challenges again. Finally, we also plan to verify different functional and non-functional properties of the set of paths to be implemented, for example, to check security / isolation issues.
%
%
\bibliographystyle{splncs04}
\bibliography{references}

\end{document}